\documentclass[10pt,journal]{IEEEtran}

\usepackage{amsmath,amssymb,amsfonts}
\usepackage{bm}
\usepackage{cite}
\usepackage{graphicx}
\usepackage{mathtools} 
\usepackage{color}
\newcommand{\C}{\mathbb{C}}
\newcommand{\R}{\mathbb{R}}
\newcommand{\E}{\mathbb{E}}
\newcommand{\diag}{\operatorname{diag}}
\newcommand{\jmathc}{\mathrm{j}}
\newcommand{\norm}[1]{\left\lVert #1\right\rVert}
\usepackage[caption=false,font=footnotesize]{subfig}
\newtheorem{proposition}{Proposition}
\newtheorem{corollary}{Corollary}

\begin{document}

\title{One-Shot Beam Tracking for
Stacked Intelligent Metasurfaces-Assisted LEO Broadcasting}

\author{Chandan~Kumar~Sheemar,~\IEEEmembership{Member,~IEEE},
        Giovanni~Iacovelli,~\IEEEmembership{Member,~IEEE}, Sourabh Solanki,~\IEEEmembership{Member,~IEEE}, Wali Ullah Khan,~\IEEEmembership{Member,~IEEE},
        and~Symeon Chatzinotas,~\IEEEmembership{Fellow,~IEEE}
\thanks{The authors are with the Interdisciplinary Centre for Security,
   Reliability and Trust (SnT), University of Luxembourg, Luxembourg
   (e-mail: \{chandan.sheemar, giovanni.iacovelli, sourabh.solanki,
   waliullah.khan, symeon.chatzinotas\}@uni.lu).}
}

\maketitle

\begin{abstract}
This letter develops a physics-informed one-shot beam tracking method
for stacked intelligent metasurface (SIM)-enabled low-Earth-orbit
(LEO) links with $K$ users and an $L$-layer SIM with $N$ meta-atoms
per layer. The predictable orbital motion is propagated through the
line-of-sight (LoS) geometry to forecast the finite-interval variation
of the user signal-to-noise ratio (SNR) profile, and an analytical
multilayer Jacobian relates this variation to the $LN$ SIM phase
coefficients. Beam tracking is formulated as normalized SNR-profile
preservation and solved through a single regularized Gauss--Newton
correction that requires one Jacobian evaluation and a $K\times K$
linear solve, thereby avoiding iterative reoptimization. A
tracking-error bound separates the regularization residual from the
nonlinear Taylor remainder and establishes a second-order local error
under a full-row-rank Jacobian. Numerical results show that the
proposed method retains approximately $94\%$ and $92\%$ of the
iterative-reoptimization rate for $L=2$ and $L=4$, respectively, while
running about $350\times$ and $260\times$ faster at $N=225$, and
confirm the predicted second-order residual error behavior.
\end{abstract}
\begin{IEEEkeywords}
Stacked intelligent metasurface, beam tracking, LEO satellite,
Gauss--Newton.
\end{IEEEkeywords}

\section{Introduction}

Low-Earth-orbit (LEO) satellite networks are emerging as a key
component of future non-terrestrial systems because they combine wide
coverage with substantially lower propagation delay than geostationary
links \cite{kumar20235g,khan2024reconfigurable}. Their high orbital
velocity, however, causes the satellite--terminal geometry to evolve
rapidly, which in turn changes the spatial beam required to maintain a
desired ground illumination pattern. For large-aperture payloads,
continuously tracking this evolution under stringent onboard processing
and latency constraints is therefore a central signal-processing
challenge.

Stacked intelligent metasurfaces (SIMs) enable high-dimensional
wave-domain beamforming through cascaded programmable layers, thereby
reducing the reliance on conventional digital RF processing
\cite{sheemar2026survey,sheemar2026energy}. This makes them attractive
for computation-constrained LEO payloads. However, the resulting spatial
transformation depends on many coupled phase variables across the SIM
layers, which must be continuously adapted as the satellite geometry
evolves. Repeatedly reoptimizing all SIM coefficients at each tracking
instant can therefore offset the very complexity advantage that motivates
wave-domain processing.

Existing work has mainly focused on SIM design for static or slowly
varying channel states. Multiuser wave-domain beamforming and joint
phase/power optimization were studied in \cite{an2025stacked}, while
\cite{lin2024stacked} considered SIM-enabled LEO communications based on
statistical channel state information (CSI) and iterative multilayer
optimization. Dynamic beam alignment has been identified as a key
challenge for SIM-enabled satellite networks \cite{zhang2026stacked}.
Predictive beamforming for highly mobile LEO links has exploited
location or state prediction \cite{xv2025predictive}, but only within
conventional array-domain architectures. None of these works provides a
sensitivity model of the multilayer SIM response to the predictable
geometry variation, nor a non-iterative rule that updates all layers
jointly.

To fill this gap, we develop a one-shot beam tracker that predicts the
future geometry and directly compensates the resulting signal variation
in the SIM control wave-domain, instead of repeatedly solving a new
high-dimensional beamforming problem. The main contributions are:
\begin{itemize}
\item Under a circular Keplerian orbit, the sampled satellite position
is shown to obey an exact two-step recurrence whose single coefficient
is fixed by the orbital rate. Propagating it through range, direction,
Earth rotation, and local-vertical/local-horizontal (LVLH) attitude
yields the finite-interval signal-to-noise-ratio (SNR) evolution without linearizing the
satellite displacement.

\item Tracking is formulated as normalized SNR-profile preservation.
The complete $K\times LN$ phase Jacobian of the multilayer SIM is
derived in closed form, and the resulting local nonlinear least-squares
problem is solved by a single regularized Gauss--Newton step that
updates all layers jointly at the cost of one Jacobian evaluation and only a
$K\times K$ solve.

\item A tracking-error bound separates the regularization residual
from the nonlinear phase-domain error. For a full-row-rank Jacobian,
the unregularized update cancels the linearized SNR mismatch, leaving a
second-order Taylor residual that is verified numerically.
\end{itemize}

The proposed framework thus exploits the predictable physics of LEO
motion to replace repeated iterative SIM reoptimization with a single
joint multilayer correction at each tracking
instant.

\emph{Notation:} lowercase and uppercase boldface letters
denote vectors and matrices; $(\cdot)^T$, $(\cdot)^H$, and
$(\cdot)^\dagger$ denote transpose, Hermitian transpose, and
pseudoinverse; $\diag(\cdot)$ is a diagonal matrix; $\Re\{\cdot\}$ is
the real part; $\|\cdot\|_2$ is the Euclidean norm; $\bm e_n$ is the
$n$-th column of $\bm I_N$; $\bm 1_K$ is the all-ones vector; and
$\mathbb{R}$ and $\mathbb{C}$ are the real and complex fields.

\section{Dynamic Signal Model}

\subsection{Multilayer SIM-Enabled Broadcasting}

\begin{figure}
    \centering
\includegraphics[width=0.8\linewidth,height=5cm,keepaspectratio]{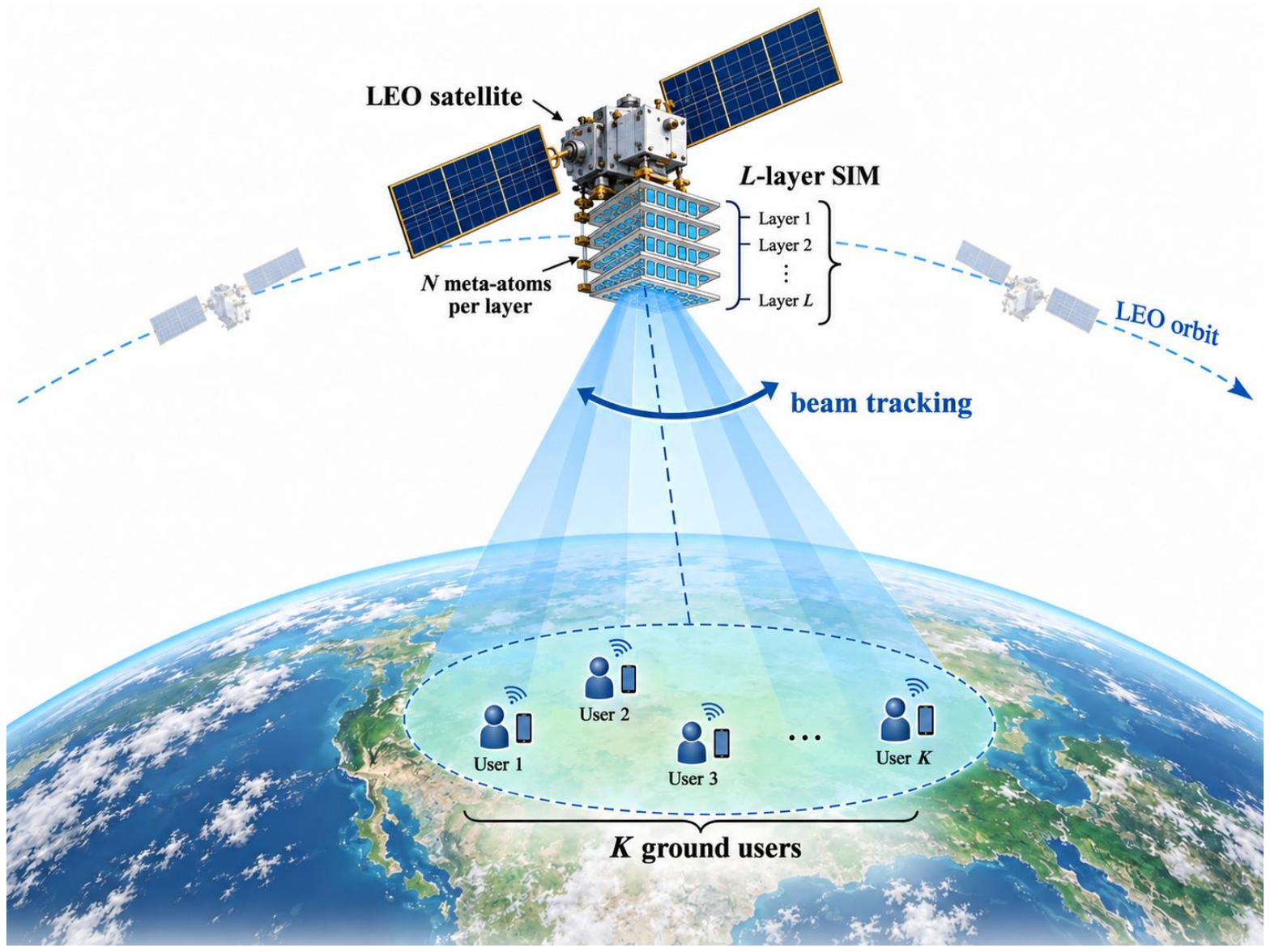}
    \caption{SIM-assisted LEO communications.}
    \label{fig:system}
\end{figure}
Consider a LEO satellite broadcasting a common unit-power symbol
$s[m]$, with $\E\{|s[m]|^2\}=1$, to $K$ single-antenna ground
terminals. A single RF feed illuminates an $L$-layer SIM composed of
$N$ independently controlled meta-atoms per layer. Let
$\bm f\in\C^N$ denote the fixed electromagnetic coupling from the feed
to the first SIM layer. The phase-control matrix of layer $l$ is
\begin{equation}
\bm\Phi_l[m]
=
\diag\!\left(
e^{\jmathc\theta_{l,1}[m]},\ldots,
e^{\jmathc\theta_{l,N}[m]}
\right).
\label{eq:Phi}
\end{equation}
Let $\bm W_l\in\C^{N\times N}$ denote the fixed propagation matrix from
layer $l-1$ to layer $l$, determined by the SIM geometry and modeled
using the Rayleigh--Sommerfeld diffraction formulation
\cite{sheemar2026survey,sheemar2026energy}. The overall SIM transfer matrix is
\begin{equation}
\bm G(\bm\theta[m])
=
\bm\Phi_L[m]\bm W_L\cdots
\bm\Phi_2[m]\bm W_2\bm\Phi_1[m],
\label{eq:G}
\end{equation}
where $ \bm\theta[m]
=
[\bm\theta_1^T[m],\ldots,\bm\theta_L^T[m]]^T
\in\R^{LN}$. Let $P$ and $\bm{h}_k$ denote the total transmit power and the channel for user $k$, respectively. The received signal at terminal $k$ is
\begin{equation}
y_k[m]
=
\sqrt P\,
\bm h_k^H[m]\bm G(\bm\theta[m])\bm f\,s[m]+n_k[m],
\label{eq:yk}
\end{equation}
where $n_k[m]\sim\mathcal{CN}(0,\sigma_k^2)$ is noise. Defining
$ 
z_k[m]
\triangleq
\sqrt P\,
\bm h_k^H[m]\bm G(\bm\theta[m])\bm f,
$
the received power and SNR are
\begin{equation}
p_k[m]=|z_k[m]|^2,
\quad
\gamma_k[m]=\frac{|z_k[m]|^2}{\sigma_k^2},
\label{eq:snr}
\end{equation}
and the common broadcast rate is
\begin{equation}
R_{\rm BC}[m]
=
\min_{k}\log_2\!\left(1+\gamma_k[m]\right).
\label{eq:bc-rate}
\end{equation}

The terminals are assumed to remain visible within the considered
tracking interval. Carrier-frequency and timing synchronization are
assumed to compensate the deterministic Doppler and delay evolution;
the present work therefore focuses on \emph{spatial} beam tracking.
Furthermore, $\bm f$ and $\{\bm W_l\}$ are fixed by the internal SIM
geometry.

\subsection{Geometry-Based LoS Channel}

Let $\bm r_s(t)$ and $\bm r_k(t)$ denote the satellite and terminal-$k$
positions in an Earth-centered inertial (ECI) frame. Define
\begin{equation}
\bm\rho_k=\bm r_k-\bm r_s,\quad
d_k=\norm{\bm\rho_k},\quad
\bm u_k=\frac{\bm\rho_k}{d_k},\quad
\bm u_k^B=\bm R_{BI}\bm u_k,
\label{eq:geometry}
\end{equation}
where $\bm R_{BI}$ transforms ECI coordinates into the satellite body
frame. Let $\bm q_n\in\R^3$ denote the position of the $n$-th meta-atom
on the final SIM layer relative to its phase center. Under LoS
propagation,
 \begin{equation}
     \bm h_k(t) =
\alpha_k(t)\bm a_k(t),\quad
[\bm a_k(t)]_n
=
e^{-\jmathc k_0\bm q_n^T\bm u_k^B(t)},
\label{eq:steering}
 \end{equation}
 
 \begin{equation}
     k_0 =\frac{2\pi}{\lambda}, \;
\alpha_k(t)
=
\sqrt{g_tg_{r,k}}\,
\frac{\lambda}{4\pi d_k(t)}
e^{-\jmathc k_0d_k(t)} ,
\label{eq:pathgain}
 \end{equation}
where $g_t$ and $g_{r,k}$ are fixed feed/element and terminal gains;
the directional transmit response is produced by $\bm G(\bm\theta)\bm f$
and is not included in $g_t$. Thus,
under the adopted LoS model, the channel is deterministically specified
by the satellite state, attitude, and terminal positions.

\subsection{Physics-Constrained Orbital Recurrence}

Consider first a circular two-body orbit with orbital radius
$R_o=R_E+h$ and angular rate $ 
\omega_o=\sqrt{\frac{\mu_E}{R_o^3}},
$
where $R_E$, $h$, and $\mu_E$ denote the Earth radius, satellite
altitude, and Earth gravitational parameter, respectively
\cite{vallado2001fundamentals}. Let $\bm p_o,\bm q_o\in\R^3$ be
orthonormal vectors spanning the orbital plane. The satellite ECI
position is
\begin{equation}
\bm r_s(t)
=
R_o\!
\left[
\bm p_o\cos(\omega_ot+\nu_0)
+
\bm q_o\sin(\omega_ot+\nu_0)
\right].
\label{eq:orbit}
\end{equation}

\begin{proposition}
Sampling \eqref{eq:orbit} at $t=mT_s$ yields the exact recurrence
\begin{equation}
\bm r_s[m+1]
=
2c_o\bm r_s[m]-\bm r_s[m-1],
\qquad
c_o\triangleq\cos(\omega_oT_s).
\label{eq:ar2}
\end{equation}
\end{proposition}

\begin{IEEEproof}
Let $\varphi_m=\omega_omT_s+\nu_0$ and $\Delta_o=\omega_oT_s$. By the
sum identities for sine and cosine,
$\bm r_s[m+1]+\bm r_s[m-1]
=
R_o\bm p_o\!\left[\cos(\varphi_m+\Delta_o)+\cos(\varphi_m-\Delta_o)\right]
+
R_o\bm q_o\!\left[\sin(\varphi_m+\Delta_o)+\sin(\varphi_m-\Delta_o)\right]
=
2\cos(\Delta_o)\bm r_s[m]$.
\end{IEEEproof}

Equation \eqref{eq:ar2} is a vector auto-regressive of second order (AR(2)) recursion whose coefficient
is dictated by the orbital dynamics rather than learned from data.
Model mismatch can be represented as
\begin{equation}
\bm r_s[m+1]
=
2c_o\bm r_s[m]-\bm r_s[m-1]+\bm\zeta_s[m],
\label{eq:perturbed}
\end{equation}
where $\bm\zeta_s[m]$ accounts for unmodeled perturbations and
ephemeris errors.

A ground terminal fixed in an Earth-centered Earth-fixed (ECEF) frame
evolves in ECI according to
\begin{equation}
\bm r_k[m+1]
=
\bm R_z(\omega_ET_s)\bm r_k[m],
\label{eq:earth}
\end{equation}
where $\omega_E$ is the Earth rotation rate. Quantities referring to
the predicted instant $m+1$ carry the superscript $+$. For a
nadir-pointing LVLH payload, define
\begin{align}
\bm h_o
&=
\frac{\bm r_s[m-1]\times\bm r_s[m]}
{\norm{\bm r_s[m-1]\times\bm r_s[m]}},
\quad
\bm e_z^+
=
-\frac{\bm r_s[m+1]}{R_o},
\nonumber\\
\bm e_x^+
&=
\frac{\bm e_z^+\times\bm h_o}
{\norm{\bm e_z^+\times\bm h_o}},
\quad
\bm e_y^+
=
\bm e_z^+\times\bm e_x^+,
\label{eq:lvlh}
\end{align}
which gives
\begin{equation}
\bm R_{BI}^+
=
[\bm e_x^+,\bm e_y^+,\bm e_z^+]^T .
\label{eq:RBI}
\end{equation}
Therefore, \eqref{eq:ar2}, \eqref{eq:earth}, and \eqref{eq:RBI}
completely determine the next LoS geometry state under the adopted circular
orbit and LVLH attitude model.

\section{Proposed One-Shot Beam Tracking}

\subsection{Finite-Interval Observation Prediction}

Define the LoS SNR mapping
\begin{equation}
\Gamma_k(\bm\xi;\bm\theta)
\triangleq
\frac{Pg_tg_{r,k}}{\sigma_k^2}
\left(\frac{\lambda}{4\pi d_k}\right)^2
\left|
\bm a_k^H(\bm u_k^B)
\bm G(\bm\theta)\bm f
\right|^2,
\label{eq:Gamma}
\end{equation}
where $\bm\xi$ collects the satellite position and attitude together
with the terminal positions, and define
$ 
\bm\Gamma(\bm\xi;\bm\theta)
=
[\Gamma_1(\bm\xi;\bm\theta),\ldots,
\Gamma_K(\bm\xi;\bm\theta)]^T$. The SIM is fully optimized only at
initialization. Specifically, $\bm\theta[0]$ is obtained from the
max--min broadcast problem
\begin{equation}
\bm\theta[0]\in
\arg\max_{\bm\theta\in[0,2\pi)^{LN}}
\min_k\log_2\!\left(1+\Gamma_k(\bm\xi[0];\bm\theta)\right).
\label{eq:init}
\end{equation}
In implementation, \eqref{eq:init} is solved by projected gradient
ascent (PGA) with Armijo backtracking on the differentiable soft-min
surrogate
\begin{equation}
\mathcal R_\tau(\bm\theta)
=
-\frac{1}{\tau}
\log
\sum_{k=1}^{K}
\exp\!\left[
-\tau\log_2(1+\Gamma_k(\bm\xi[0];\bm\theta))
\right],
\label{eq:softmin}
\end{equation}
where larger $\tau$ gives a tighter approximation of the minimum rate.
The same optimizer is subsequently used for the full-reoptimization
benchmark. Let $ 
\bm\gamma^{\rm ref}
\triangleq
\bm\Gamma(\bm\xi[0];\bm\theta[0])$ and
$\bm D_\gamma
\triangleq
\diag(\bm\gamma^{\rm ref})$, and define the normalized SNR profile
\begin{equation}
\overline{\bm\Gamma}(\bm\xi;\bm\theta)
\triangleq
\bm D_\gamma^{-1}\bm\Gamma(\bm\xi;\bm\theta).
\label{eq:normalizedGamma}
\end{equation}
Thus, the desired footprint corresponds to
$\overline{\bm\Gamma}=\bm 1_K$. The normalization prevents high-SNR
terminals from dominating the tracking objective and gives equal
weight to the \emph{relative} SNR variation of every terminal.

At tracking instant $m$, \eqref{eq:ar2} produces
\begin{equation}
\bm r_s^+[m]
=
2c_o\bm r_s[m]-\bm r_s[m-1],
\label{eq:rplus}
\end{equation}
which, jointly with \eqref{eq:earth} and \eqref{eq:RBI}, determines
the future geometry $\bm\xi^+[m]$. Keeping the SIM fixed at
$\bm\theta[m]$, the normalized SNR variation over the tracking
interval is
\begin{equation}
\overline{\bm\delta}^{\rm obs}[m]
=
\overline{\bm\Gamma}(\bm\xi^+[m];\bm\theta[m])
-
\overline{\bm\Gamma}(\bm\xi[m];\bm\theta[m]).
\label{eq:forcing}
\end{equation}
Since \eqref{eq:ar2}, \eqref{eq:earth}, and \eqref{eq:RBI} give
$\bm\xi^+[m]$ exactly under the adopted models and \eqref{eq:Gamma} is
a deterministic mapping of the geometry, \eqref{eq:forcing} is exact
and involves no linearization of the orbital displacement.

Define the normalized accumulated tracking error
\begin{equation}
\overline{\bm\varepsilon}[m]
=
\overline{\bm\Gamma}(\bm\xi[m];\bm\theta[m])
-\bm 1_K .
\label{eq:tracking-error}
\end{equation}
The mismatch at the subsequent geometry \emph{before} SIM
reconfiguration is consequently
\begin{equation}
\overline{\bm b}[m]
=
\overline{\bm\varepsilon}[m]
+
\overline{\bm\delta}^{\rm obs}[m]
=
\overline{\bm\Gamma}(\bm\xi^+[m];\bm\theta[m])
-\bm 1_K .
\label{eq:b-mismatch}
\end{equation}
Hence, each update compensates both the newly predicted variation and
any residual error accumulated in preceding intervals.

\subsection{Multilayer Phase Sensitivity}

Factor the SIM transfer matrix around layer $l$ as
$ \bm G=\bm V_l\bm\Phi_l\bm U_l$,
where
\begin{align}
\bm U_l
&=
\begin{cases}
\bm I_N, & l=1,\\
\bm W_l\bm\Phi_{l-1}\cdots\bm W_2\bm\Phi_1, & l>1,
\end{cases}
\nonumber\\[-1mm]
\bm V_l
&=
\begin{cases}
\bm I_N, & l=L,\\
\bm\Phi_L\bm W_L\cdots\bm\Phi_{l+1}\bm W_{l+1}, & l<L .
\end{cases}
\label{eq:UV}
\end{align}
Since $ 
\frac{\partial\bm\Phi_l}{\partial\theta_{l,n}}
=
\jmathc e^{\jmathc\theta_{l,n}}
\bm e_n\bm e_n^T$, 
the derivative of the complete multilayer transfer function is
\begin{equation}
\frac{\partial\bm G}{\partial\theta_{l,n}}
=
\jmathc e^{\jmathc\theta_{l,n}}
\bm V_l\bm e_n\bm e_n^T\bm U_l .
\label{eq:dG}
\end{equation}
At the predicted geometry, let $
z_k^+(\bm\theta)
=
\sqrt P\,(\bm h_k^+)^H\bm G(\bm\theta)\bm f$.
Then
\begin{equation}
\frac{\partial z_k^+}{\partial\theta_{l,n}}
=
\jmathc\sqrt P\,e^{\jmathc\theta_{l,n}}
\left((\bm h_k^+)^H\bm V_l\bm e_n\right)
\left(\bm e_n^T\bm U_l\bm f\right),
\label{eq:dz}
\end{equation}
and
\begin{equation}
\frac{\partial\Gamma_k^+}{\partial\theta_{l,n}}
=
\frac{2}{\sigma_k^2}
\Re\!\left\{
(z_k^+)^*
\frac{\partial z_k^+}{\partial\theta_{l,n}}
\right\}.
\label{eq:dgamma}
\end{equation}
Collecting \eqref{eq:dgamma} for every user, layer, and meta-atom gives
\begin{equation}
\bm J_\theta^+[m]
\triangleq
\left.
\frac{\partial
\bm\Gamma(\bm\xi^+[m];\bm\theta)}
{\partial\bm\theta^T}
\right|_{\bm\theta=\bm\theta[m]}
\in\R^{K\times LN}.
\label{eq:Jraw}
\end{equation}
Correspondingly, the normalized phase Jacobian is
\begin{equation}
\overline{\bm J}_\theta^+[m]
\triangleq
\bm D_\gamma^{-1}\bm J_\theta^+[m].
\label{eq:J}
\end{equation}
Thus, one $K\times LN$ matrix captures the sensitivity of the whole
normalized footprint to all $LN$ coupled phases.

\subsection{One-Shot Nonlinear Correction}

At $\bm\xi^+[m]$, the desired configuration should recover the
reference normalized footprint. Ideally,
\begin{equation}
\overline{\bm\Gamma}
\left(
\bm\xi^+[m];
\bm\theta[m]+\Delta\bm\theta[m]
\right)
=
\bm 1_K .
\label{eq:ideal}
\end{equation}
Because $\bm G(\bm\theta)$ contains cascaded products of $L$
phase-dependent matrices, \eqref{eq:ideal} is a nonlinear system in
$LN$ coupled variables. Solving it iteratively at every tracking
instant would largely eliminate the computational advantage obtained
from predicting the orbital evolution.

We instead seek a small phase correction that restores the normalized
SNR profile by considering the regularized nonlinear least-squares
problem
\begin{equation}
\begin{aligned}
\min_{\Delta\bm\theta}\quad&
\left\|
\overline{\bm\Gamma}
\left(
\bm\xi^+[m];
\bm\theta[m]+\Delta\bm\theta
\right)
-\bm 1_K
\right\|_2^2
+
\eta\norm{\Delta\bm\theta}_2^2 ,
\end{aligned}
\label{eq:nonlinear-problem}
\end{equation}
where $\eta>0$ regularizes poorly conditioned phase sensitivities
and penalizes excessive SIM reconfiguration.

A first-order expansion around $\bm\theta[m]$ gives
\begin{equation}
\begin{split}
&
\overline{\bm\Gamma}
\left(
\bm\xi^+[m];
\bm\theta[m]+\Delta\bm\theta
\right)
-\bm 1_K
=
\overline{\bm b}[m]
+
\overline{\bm J}_\theta^+[m]\Delta\bm\theta
+
\bm r_\theta(\Delta\bm\theta),
\end{split}
\label{eq:GN-full}
\end{equation}
where $\bm r_\theta(\Delta\bm\theta)$ contains the second- and
higher-order phase terms. Neglecting this remainder yields one damped
Gauss--Newton, equivalently Levenberg--Marquardt, step:
\begin{equation}
\min_{\Delta\bm\theta}
\;
\norm{
\overline{\bm J}_\theta^+[m]\Delta\bm\theta+
\overline{\bm b}[m]
}_2^2
+
\eta\norm{\Delta\bm\theta}_2^2 .
\label{eq:LS}
\end{equation}

\begin{proposition}[One-shot multilayer update]
For $\eta>0$, \eqref{eq:LS} is strictly convex and has the unique
solution
\begin{equation}
\Delta\bm\theta^\star[m]
=
-
(\overline{\bm J}_\theta^+)^T
\left[
\overline{\bm J}_\theta^+
(\overline{\bm J}_\theta^+)^T
+
\eta\bm I_K
\right]^{-1}
\overline{\bm b}[m].
\label{eq:update}
\end{equation}
\end{proposition}

 \begin{IEEEproof}
The objective of \eqref{eq:LS} is a quadratic whose Hessian,
$2[(\overline{\bm J}_\theta^+)^T\overline{\bm J}_\theta^+
+\eta\bm I_{LN}]$, satisfies
$\bm x^T[(\overline{\bm J}_\theta^+)^T\overline{\bm J}_\theta^+
+\eta\bm I_{LN}]\bm x
=\norm{\overline{\bm J}_\theta^+\bm x}_2^2+\eta\norm{\bm x}_2^2>0$
for every $\bm x\neq\bm 0$ and $\eta>0$. Hence \eqref{eq:LS} is
strictly convex and its unique minimizer solves the stationarity
condition $
\left[
(\overline{\bm J}_\theta^+)^T\overline{\bm J}_\theta^+
+\eta\bm I_{LN}
\right]\Delta\bm\theta^\star
=
-(\overline{\bm J}_\theta^+)^T\overline{\bm b}$. Since $(\bm A^T\bm A+\eta\bm I_{LN})\bm A^T
=\bm A^T(\bm A\bm A^T+\eta\bm I_K)$ for any
$\bm A\in\R^{K\times LN}$, with both parenthesized matrices invertible
for $\eta>0$, one has
$(\bm A^T\bm A+\eta\bm I_{LN})^{-1}\bm A^T
=\bm A^T(\bm A\bm A^T+\eta\bm I_K)^{-1}$.   Applying this identity with $\bm A=\overline{\bm J}_\theta^+$ to the
   solution of the stationarity condition yields \eqref{eq:update}.
\end{IEEEproof}

Although $\Delta\bm\theta^\star\in\R^{LN}$ simultaneously updates
every meta-atom of every SIM layer, \eqref{eq:update} inverts only a
$K\times K$ matrix and therefore avoids both an $LN\times LN$
inversion and layer-by-layer alternating optimization; the dominant
cost is a single Jacobian evaluation (Sec.~III-E). Moreover, $\eta$
bounds the phase excursion: the singular values $\sigma_k$ of
$\overline{\bm J}_\theta^+$ enter \eqref{eq:update} through the
factors $\sigma_k/(\sigma_k^2+\eta)\leq1/(2\sqrt{\eta})$, so that
\begin{equation}
\norm{\Delta\bm\theta^\star[m]}_2
\leq
\frac{\norm{\overline{\bm b}[m]}_2}{2\sqrt{\eta}} .
\label{eq:excursion}
\end{equation}
Decreasing $\eta$ thus prioritizes footprint restoration, whereas
increasing it produces smoother and smaller phase excursions, and
\eqref{eq:excursion} bounds the nonlinear remainder of Sec.~III-D by
$\beta\norm{\overline{\bm b}[m]}_2^2/(8\eta)$.

 For ideal continuous-phase hardware,
\begin{equation}
\bm\theta[m+1]
=
\mathcal W_{2\pi}
\left(
\bm\theta[m]+\Delta\bm\theta^\star[m]
\right),
\label{eq:phase-update}
\end{equation}
where $\mathcal W_{2\pi}(\cdot)$ wraps every entry to $[0,2\pi)$.
With a maximum admissible phase change $\Delta_{\max}$,
\begin{equation}
\bm\theta[m+1]
=
\mathcal W_{2\pi}
\left[
\bm\theta[m]+
\mathcal C_{\Delta_{\max}}
(\Delta\bm\theta^\star[m])
\right],
\label{eq:bounded}
\end{equation}
where $\mathcal C_{\Delta_{\max}}(\cdot)$ clips each entry to
$[-\Delta_{\max},\Delta_{\max}]$. For a $B_\phi$-bit implementation,
\eqref{eq:bounded} is subsequently projected onto the discrete set
\begin{equation}
\mathcal Q_{B_\phi}
=
\left\{
0,\frac{2\pi}{2^{B_\phi}},\ldots,
\frac{2\pi(2^{B_\phi}-1)}{2^{B_\phi}}
\right\}.
\label{eq:quant}
\end{equation}

Combining \eqref{eq:ar2}, \eqref{eq:forcing}, and \eqref{eq:update}
gives the complete trajectory-to-SIM relation
\begin{equation}
\{\bm r_s[m-1],\bm r_s[m]\}
\rightarrow
\bm\xi^+[m]
\rightarrow
\overline{\bm b}[m]
\rightarrow
\Delta\bm\theta^\star[m]
\rightarrow
\bm\theta[m+1].
\label{eq:chain}
\end{equation}
Thus, two initial ephemeris positions suffice: in open-loop operation,
$\bm r_s[m-1]$ and $\bm r_s[m]$ are themselves earlier predictions.
\subsection{Tracking Accuracy}

Under the adopted orbital model, the approximation arises only from
the first-order phase linearization in \eqref{eq:GN-full}.

\begin{proposition}[Tracking-error bound]
Let $\overline{\bm J}_\theta(\bm\theta)$ denote the normalized
Jacobian at $\bm\xi^+[m]$, and assume it is $\beta$-Lipschitz, i.e.,
\begin{equation}
\norm{
\overline{\bm J}_\theta(\bm\theta_1)
-
\overline{\bm J}_\theta(\bm\theta_2)
}_2
\leq
\beta\norm{\bm\theta_1-\bm\theta_2}_2 ,
\label{eq:lipschitz}
\end{equation}
for all $\bm\theta_1,\bm\theta_2$ on the segment between $\bm\theta[m]$
and $\bm\theta[m]+\Delta\bm\theta^\star[m]$. Then
\begin{equation}
\begin{split}
&
\norm{
\overline{\bm\Gamma}
(\bm\xi^+;\bm\theta+\Delta\bm\theta^\star)
-\bm 1_K
}_2
\\
&\leq
\eta
\norm{
\left[
\overline{\bm J}_\theta^+
(\overline{\bm J}_\theta^+)^T
+\eta\bm I_K
\right]^{-1}
\overline{\bm b}
}_2
+
\frac{\beta}{2}
\norm{\Delta\bm\theta^\star}_2^2 .
\end{split}
\label{eq:error-bound}
\end{equation}
\end{proposition}

\begin{IEEEproof}
Taylor's theorem gives
$\norm{\bm r_\theta(\Delta\bm\theta^\star)}_2
\leq
\frac{\beta}{2}\norm{\Delta\bm\theta^\star}_2^2$.
Moreover, from \eqref{eq:update},
\begin{equation}
\overline{\bm b}
+
\overline{\bm J}_\theta^+\Delta\bm\theta^\star
=
\eta
\left[
\overline{\bm J}_\theta^+
(\overline{\bm J}_\theta^+)^T
+\eta\bm I_K
\right]^{-1}
\overline{\bm b}.
\label{eq:linear-residual}
\end{equation}
The result follows by the triangle inequality.
\end{IEEEproof}

The two terms in \eqref{eq:error-bound} respectively quantify the
regularization residual and the nonlinear phase error, the latter being
second order in $\norm{\Delta\bm\theta^\star}_2$.

\begin{corollary}
If $\overline{\bm J}_\theta^+$ has full row rank (which requires
$LN\geq K$) and $\eta\rightarrow0$, then
\begin{equation}
\Delta\bm\theta^\star
=
-
(\overline{\bm J}_\theta^+)^\dagger
\overline{\bm b},
\quad
\overline{\bm b}
+
\overline{\bm J}_\theta^+\Delta\bm\theta^\star
=
\bm 0.
\label{eq:pseudo}
\end{equation}
Hence, the minimum-norm update cancels the linearized normalized SNR
mismatch exactly, leaving only a second-order phase-domain residual.
\end{corollary}

Since SIMs typically satisfy $LN\gg K$, this full-row-rank regime is
particularly relevant. The above results apply to unconstrained
continuous phases; clipping and finite-resolution quantization introduce
additional hardware-projection errors.

  \subsection{Implementation and Complexity}

The Jacobian in \eqref{eq:Jraw} is efficiently computed by reusing
forward and backward fields:
$\bm x_1=\bm f$ and
$\bm x_l=\bm W_l\bm\Phi_{l-1}\bm x_{l-1}$, $l=2,\ldots,L$,
while, for user $k$,
\begin{equation}
(\bm v_{k,L})^H=(\bm h_k^+)^H,\quad
(\bm v_{k,l})^H=
(\bm v_{k,l+1})^H\bm\Phi_{l+1}\bm W_{l+1}.
\label{eq:backward}
\end{equation}
Then,
$
\frac{\partial z_k^+}{\partial\theta_{l,n}}
=
\jmathc\sqrt P\,e^{\jmathc\theta_{l,n}}
[\bm v_{k,l}^*]_n[\bm x_l]_n
$,
so all $KLN$ sensitivities follow from the same recursions. For dense
inter-layer matrices, the forward/backward passes require
$\mathcal O((K+1)LN^2)$ operations, forming
$\overline{\bm J}_\theta^+
(\overline{\bm J}_\theta^+)^T$ requires
$\mathcal O(K^2LN)$, and the $K\times K$ inversion costs
$\mathcal O(K^3)$. Hence, each one-shot update requires only
\begin{equation}
\mathcal O\!\left((K+1)LN^2+K^2LN+K^3\right),
\label{eq:complexity}
\end{equation}
i.e., one gradient-type pass, which iterative reoptimization repeats
at every iteration.

\section{Numerical Results}

The simulation parameters are summarized in Table~\ref{tab:sim}. The
SIM coupling vector $\bm f$ and inter-layer matrices $\{\bm W_l\}$
follow the Rayleigh--Sommerfeld model. The ground-truth orbit is
numerically propagated with the dominant $J_2$ perturbation
\cite{vallado2001fundamentals}, whereas the tracker is initialized from
two ephemeris positions and thereafter propagates \eqref{eq:ar2} in open
loop from its own predictions, so that the predictor is never evaluated
on the ideal trajectory from which it is derived. The reported results
therefore include the accumulated orbital model mismatch and are
averaged over $50$ independent terminal/channel realizations. Runtimes
are measured in MATLAB R2024b on an Intel Core i7-12700 at 2.1 GHz.

 \begin{table}[t]
\centering
\caption{Simulation Parameters}
\label{tab:sim}
\scriptsize
\setlength{\tabcolsep}{0.75pt}
\renewcommand{\arraystretch}{0.75}
\resizebox{\columnwidth}{!}{%
\begin{tabular}{l c l c}
\hline
Parameter & Value & Parameter & Value\\
\hline
Altitude $h$ & $1000$ km &
Carrier $f_c$ & $4$ GHz\\
Bandwidth $B$ & $50$ MHz &
Transmit power $P$ & $30$ dBW\\
Earth radius $R_E$ & $6371$ km &
$\mu_E$ & $3.986\times10^{14}$ m$^3$/s$^2$\\
Feed gain $g_t$ & $6$ dBi &
Terminal gain $g_{r,k}$ & $0$ dBi\\
Noise temperature & $290$ K &
Number of users $K$ & $9$\\
Coverage diameter & $600$ km &
Inclination & $53^\circ$\\
Element spacing & $\lambda/2$ &
Tracking interval $T_s$ & $1$ s\\
Phase resolution $B_\phi$ & $6$ bits &
$\Delta_{\max}$ & $\pi/8$\\
Regularization $\eta$ & $10^{-3}$ &
Soft-min parameter $\tau$ & $10$\\
PGA tolerance / max. iter. & $10^{-5}$ / $500$ &
Realizations & $50$\\
\hline
\end{tabular}%
}
\end{table}

For the tracking and local-error experiments, we consider the two SIM
configurations $(L,N)=(2,100)$ and $(4,225)$; the runtime experiment
additionally varies $N\in\{36,64,100,144,225\}$ for $L\in\{2,4\}$.
Three beam-control strategies are compared. \emph{Fixed SIM} retains
the initial configuration $\bm\theta[0]$; \emph{proposed one-shot
tracking} uses the predicted geometry and \eqref{eq:update}; and
\emph{full reoptimization} jointly updates all $LN$ phases through PGA
on \eqref{eq:softmin} with Armijo backtracking from the ground-truth
LoS geometry, warm-started from the preceding configuration and
iterated until convergence, as a high-complexity reference. The same
phase resolution and phase-increment constraints
are applied whenever relevant.

\begin{figure}[t]
    \centering

    \subfloat[Runtime versus $N$.]
    {\includegraphics[width=0.48\columnwidth]{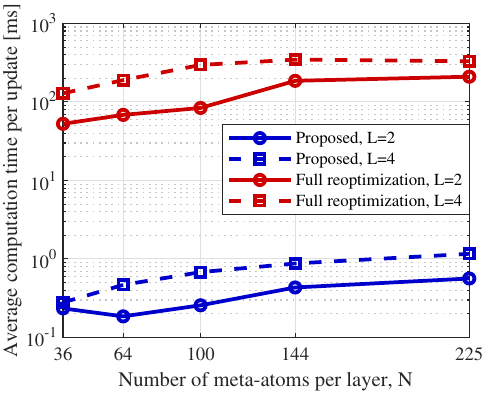}
    \label{fig:runtime}}
    \hfill
    \subfloat[Taylor residual versus phase perturbation.]
    {\includegraphics[width=0.48\columnwidth]{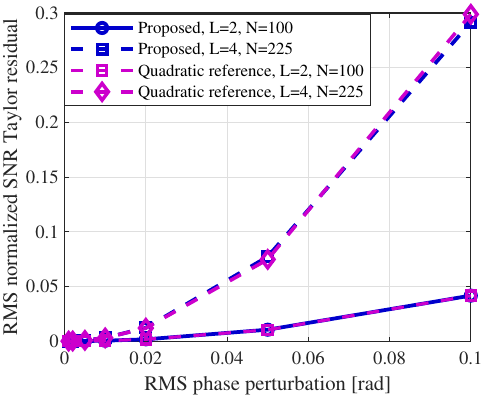}
    \label{fig:residual}}

    \caption{Computational complexity and local approximation accuracy
    of the proposed one-shot SIM tracker.} 
    \label{fig:complexity_accuracy}
\end{figure}

\begin{figure}[t]
    \centering
    \includegraphics[width=0.5\columnwidth]{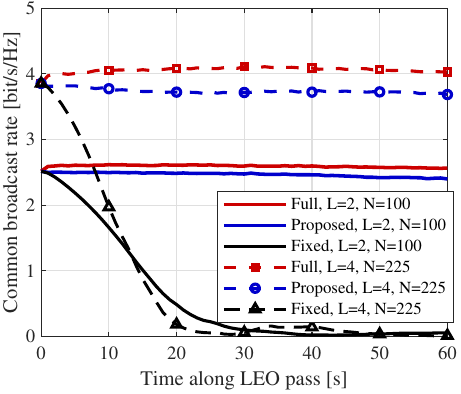}
    \caption{Common broadcast rate along the LEO pass for the fixed
    SIM, proposed one-shot tracker, and full SIM reoptimization with
    $(L,N)=(2,100)$ and $(4,225)$.}
    \label{fig:rate}
\end{figure}

 Fig.~\ref{fig:complexity_accuracy}(a) shows that the proposed one-shot
update is substantially faster than iterative SIM reoptimization over
the entire range of $N$. At $N=225$, the proposed method requires
approximately $0.7$ and $1.4$~ms for $L=2$ and $L=4$, respectively,
compared with about $230$ and $370$~ms for full reoptimization,
corresponding to roughly $350\times$ and $260\times$ lower runtimes.
The computational advantage is therefore preserved as both the number
of meta-atoms and SIM layers increase.
Fig.~\ref{fig:complexity_accuracy}(b) confirms the predicted
second-order local behavior: for both SIM configurations, the measured
Taylor residual closely follows a quadratic in the perturbation
magnitude, as predicted by the second term of \eqref{eq:error-bound},
over the considered phase-perturbation range.

Finally, Fig.~\ref{fig:rate} shows that the fixed SIM rapidly loses beam
alignment, whereas the proposed tracker maintains an almost constant
common rate throughout the pass. At $t=60$~s, the proposed method
achieves approximately $2.4$ and $3.7$~bit/s/Hz for
$(L,N)=(2,100)$ and $(4,225)$, respectively, compared with about
$2.55$ and $4.0$~bit/s/Hz under full reoptimization. Hence, the
one-shot tracker retains approximately $94\%$ and $92\%$ of the
fully reoptimized rate while avoiding iterative optimization. The
remaining gap stems from the fixed reference profile
$\bm\gamma^{\rm ref}$, which drifts from the max--min optimum as the
path losses evolve, from the first-order phase update, and from
clipping and quantization. 
\section{Conclusion}

This letter proposed a physics-informed one-shot tracking method for
SIM-enabled LEO links. Orbital prediction and an analytical multilayer
Jacobian are combined with a regularized Gauss--Newton update that
jointly adjusts all $LN$ SIM phases at the cost of one Jacobian
evaluation and a $K\times K$ solve. The derived error analysis
establishes a second-order local residual, while numerical results
demonstrate effective tracking with substantially lower complexity
than full SIM reoptimization.

\bibliographystyle{IEEEtran}
\bibliography{main}

\end{document}